\documentclass[12pt]{article}

\usepackage[left=3cm,top=3cm,right=3cm,bottom=3cm]{geometry}

\usepackage{graphicx}
\usepackage{latexsym}
\usepackage{verbatim}
\usepackage{tikz}
\usepackage{array}
\usepackage{amssymb,amsmath}
\usepackage{amsthm}
\usepackage{cite}
\usepackage{chngcntr}
\counterwithout{table}{section}
\counterwithout{figure}{section}

\usepackage{url}
\urldef{\mailsa}\path|{jeannette.janssen,rory.ross.wilson}@dal.ca|  
\urldef{\mailsb}\path|pralat@ryerson.ca|  

\newcommand{\wep}{{\sl{w.e.p.}}}
\newcommand{\aas}{{\sl{a.a.s.}}}

\newtheorem{theorem}{Theorem}[section]
\newtheorem{lemma}[theorem]{Lemma}
\newtheorem{cor}[theorem]{Corollary}

\newcommand{\Prob}{\mathbb{P}}
\newcommand{\E}{\mathbb E}
\newcommand\eps{\varepsilon}

\newcommand{\reg}{\mathcal{R}}

\newcommand{\pp}{\mathbb P}

\providecommand{\keywords}[1]{\textbf{\textit{Keywords---}} #1}

\begin{document}

\title{Non-Uniform Distribution of Nodes in the Spatial Preferential Attachment Model}

\author{
Jeannette Janssen\thanks{Department of Mathematics and Statistics, Dalhousie University, Halifax, NS, Canada} 
\and 
Pawe\l{}~Pra\l{}at\thanks{Department of Mathematics, Ryerson University, Toronto, ON, Canada}
\and 
Rory Wilson\thanks{Department of Mathematics and Statistics, Dalhousie University, Halifax, NS, Canada}  
}


\maketitle

\begin{abstract}
The spatial preferential attachment (SPA) is a model for complex networks. In the SPA model, nodes are embedded in a metric space, and each node has a sphere of influence whose size increases if the node gains an in-link, and otherwise decreases with time. In this paper, we study the behaviour of the SPA model when the distribution of the nodes is non-uniform. Specifically, the space is divided into dense and sparse regions, where it is assumed that the dense regions correspond to coherent communities. We prove precise theoretical results regarding the degree of a node, the number of common neighbours, and the average out-degree in a region.  Moreover, we show how these theoretically derived results about the graph properties of the model can be used to formulate a reliable estimator for the distance between certain pairs of nodes, and to estimate the density of the region containing a given node.
\end{abstract}

\keywords{Spatial Random Graphs, Spatial Preferrential Attachment Model, Preferential Attachment, Complex Networks, Web Graph, Co-citation, Common Neighbours}

\section{Introduction}

There has been a great deal of recent interest in modelling complex networks, a result of the increasing connectedness of our world. The hyperlinked structure of the Web, citation patterns, friendship relationships, infectious disease spread, these are seemingly disparate linked data sets which have fundamentally very similar natures. 

Many models of complex networks have a common weakness: the `uniformity' of the nodes; other than link structure there is no way to distinguish the nodes. One family of models which overcomes this deficiency is the family of spatial (or geometric) models, wherein the nodes are embedded in a metric space. A node's position---especially in relation to the others---has real-world meaning: the character of the node is encoded in its location. Similar nodes are closer in the space than dissimilar nodes. This metric space has many potential meanings: in communication networks, perhaps physical distance; in a friendship graph, an interest space; in the World Wide Web, a topic space. As an illustration, a node representing a webpage on pet food would be closer in the metric space to one on general pet care than to one on travel. 

The Spatial Preferrential Attachment Model~\cite{spa1}, designed as a model for the World Wide Web, is one such spatial model. Indeed, as its name suggests, the SPA Model combines geometry and preferential attachment. Setting the SPA Model apart is the incorporation of `spheres of influence' to accomplish preferential attachment: the greater the degree of the node, the larger its sphere of influence, and hence the higher the likelihood of the node gaining more neighbours. The SPA model produces scale-free networks, which exhibit many of the characteristics of real-life networks (see~\cite{spa1,typSPA}). In~\cite{matt}, it was shown that the SPA model gave the best fit, in terms of graph structure, for a series of social networks derived from Facebook. 

As the motivation behind spatial models is the `second layer of meaning'---the character of the nodes as represented by their positions in the metric space---we hope to uncover this layer through examination of the link structure. In particular, estimating the distance between nodes in the metric space forms the basis for two important link mining tasks: finding entities that are similar---represented by nodes that are close together in the metric space---and finding communities---represented by spatial clusters of nodes in the metric space. We show how a theoretical analysis of a spatial model can lead to reliable tools to extract the `second layer of meaning'.

The majority of the spatial models to this point have used uniform random distribution of nodes in the space. However, considering the real-world networks these models represent, this concept does not capture the following essential aspect of real-life data. Indeed, on a basic level, if the metric space represents actual physical space, and the nodes people, then we note that people cluster in cities and towns, rather than being uniformly spread across the land. More abstractly, there are more webpages on a popular topic, corresponding to a small area of our metric space, than for a more obscure topic. The development of spatial network models naturally then begins to incorporate varying densities of node distribution: both `clumps' of higher/lower density, as well as gradually changing densities, are both possibilities. Of the more important goals is that of community recognition: the discovery and quantification of characteristically (semantically) similar nodes. 

In this work we generalize the SPA model to an inhomogeneous distribution of nodes within the space. We assume distinct regions of different densities, where the dense regions are the `clusters'. We find that the local regions behave almost as if generated by independent SPA models of parameters derived from the densities. Many earlier results from the SPA Model then translate easily to this inhomogeneous version and we begin the process of uncovering the geometry using link analysis.

In the remainder of this section, we first review related work, and then we give a formal definition of the SPA model. In Section~\ref{sec:graphprop} we state our main theoretical results. In particular, we give the typical behaviour of the in-degree of a node, and use this to derive a relationship between spatial distance and number of common neighbours of a pair of nodes. The proofs of the theorems are given in Section~\ref{sec:proofs}. 

In Section~\ref{sec:Sim} we verify the asymptotic results from Section~\ref{sec:graphprop} through a simulation of the SPA model to generate large graphs. Specifically, we show how the relationship between spatial distance and common neighbours can be used to devise a distance estimator which gives precise results. We also use the theoretical results to estimate the local density around a node. Our simulations show that these estimators give reliable results on the simulated data.
 
\subsection{Background and Related Work}

Efforts to extract node information through link analysis began with a heuristic quantification of entity similarity: numerical values, obtained from the graph structure, indicating the relatedness of two nodes. Early simple  measures of entity similarity, such as the Jaccard coefficient~\cite{jaccard1}, gave way to iterative graph theoretic measures, in which two objects are similar if they are related to similar objects, such as SimRank~\cite{simrank}. Many such measures also incorporate co-citation, the number of common neighbours of two nodes, as proposed in the context of bibliographic research in an early paper by Small~\cite{small}. In~\cite{Hoff}, the authors make inferences on the social space for nodes in a social network, using Bayesian methods and maximum likelihood.

Generative spatial models were proposed in a more general setting, where the main objective was to generate graphs with properties that correspond to those observed in real-life networks. Different approaches were explored, for example in~\cite{threshbrad} using thresholds, or in~\cite{geopref1, geopref2} using a geometric variant of the preferential attachment. Graph properties of this model were analyzed by Jordan in~\cite{jordan_seq}; follow-up work on this model can be found in~\cite{jordanwade}. In~\cite{jordan_nonuni}, a non-uniform distribution of the points in space is considered. In~\cite{jacob1}, Jacob and M\"{o}rters propose a probabilistic spatial model where the link probability is a function decreasing with distance. The setting is general, and includes the SPA model as a special case. Follow-up work on this model can be found in~\cite{jacob3}. 

The SPA model was first proposed in~\cite{spa1} as a model for the World Wide Web. In~\cite{spa1} and~\cite{typSPA}, it was proved that the SPA model produces graphs with certain graph properties that correspond to those observed in real-life networks.  The authors' previous paper,~\cite{geo1}, used common neighbours to explore the underlying geometry of the SPA model and quantify node similarity based on distance in the space.  However, the distribution of nodes in space was assumed to be uniform. The approach used in this paper is similar to that in~\cite{geo1}, but we investigate the complications that arise when the distribution is non-uniform, which is clearly a more realistic setting.

An earlier version of this work, containing no proofs, was presented at the workshop WAW 2013. An extended abstract can be found in~\cite{waw2013version}.

\subsection{The Inhomogeneous SPA Model}

We begin with a brief description of our inhomogeneous SPA model. The model presented here is a generalization of the SPA model introduced in~\cite{spa1}, the main difference being that we allow for an inhomogeneous distribution of nodes in the space. 

Let $S$ be the unit hypercube in $\mathbb{R}^m$, equipped with the torus metric derived from the Euclidean norm, or any equivalent metric. The nodes $\{ v_t\}_{t=1}^n$ of the graphs produced by the SPA model are points in $S$ chosen via an $m$-dimensional point process.  Most generally, the process is given by a probability density function $\rho$; $\rho$ is a measurable function such that $\int_{S} \rho d\mu = 1$.  Precisely, for any measurable set $A\subseteq S$ and any $t$ such that $1 \le t \le n$, $\pp(v_t \in A) = \int_{A} \rho d\mu.$

In fact, we will restrict ourselves to probability functions that are {\sl locally constant}. Precisely, we assume that the space $S=[0,1)^m$ is divided into $k^m$ equal sized hypercubes, where $k$ is a constant natural number. Each hypercube is of the form $I_{j_1}\times I_{j_2}\times \dots \times I_{j_m}$ ($0 \le j_1, j_2, \ldots, j_m < k$), where $I_{j}=[j/k,(j+1)/k)$.  Note that any density function $\rho$ can be approximated by such a locally constant function, so that this restriction is justified. 

To keep notation as simple as possible, we assume that each hypercube is labelled $\reg_\ell$, $1\leq \ell\leq k^m$. Let $\rho_\ell$ be the density of $\reg_\ell$, so the density function has value $\rho_\ell$ on $\reg_\ell$. For any node $v$, let $\reg(v)$ be the hypercube containing $v$, and let $\rho(v)$ be the density of $\reg(v)$. Clearly, every hypercube has volume $k^{-m}$. Then the probability that a node $v_t$, introduced at time $t$, falls in $\reg_\ell$ equals $q_\ell = \rho_\ell k^{-m}$, and the expected number of points in $\reg_\ell$ equals $q_\ell n = \rho_\ell k^{-m} n$. It is easy to see that $\sum_\ell q_\ell=1$. Thus we may model the point process as follows: at each time step $t$, one of the regions is chosen as the destination of $v_t$; region $\reg_\ell$ is chosen with probability $q_\ell$. Then, a location for $v_t$ is chosen uniformly at random from the chosen region $\reg_\ell$.

The SPA model generates stochastic sequences for graphs $\{G_t\}_{t\geq 0}$; for each $t \ge 0$, $G_t=(V_t, E_t)$, where $E_t$ is an edge set, and $V_t \subseteq S$ is a node set. The in-degree of a node $v$ at time $t$ is given by $\deg^-(v,t)$. Likewise the out-degree is given by $\deg^+(v,t)$. The sphere of influence $S(v,t)$ of a node $v$ at time $t$ is defined as the ball, centred at $v$, with total volume
$$
|S(v,t)| = \frac{A_1 \deg^-(v,t) + A_2}{t},
$$
where $A_1, A_2>0$ are given parameters. If $(A_1 \deg^-(v,t) + A_2)/t \geq 1$, then $S(v,t) = S$ and so $|S(v,t)| = 1$. We impose the additional restriction that $pA_1\max_j \rho_j < 1$; this avoids regions becoming too dense. This property will be always assumed. The generation of a SPA model graph begins at time $t = 0$ with $G_0$ being the null graph. At each time step $t \geq 1$ (defined to be the transition from $G_{t-1}$ to $G_t$), a node $v_t$ is chosen from $S$ according to the given spatial distribution, and added to $V_{t-1}$ to form $V_t$.  Next, independently, for each node $u\in V_{t-1}$ such that $v_t \in S(u,t-1)$, a directed link $(v_{t},u)$ is created with probability $p$, $p \in (0,1)$ being another parameter of the model. 

Let $\delta(v)$ be the distance from $v$ to the boundary of $\reg(v)$.  Let $r(v,t)$ be the radius of the sphere of influence of node $v$ at time $t$. So if $r(v,t)\leq \delta(v)$, then $S(v,t)$ is completely contained in $\reg(v)$ at time $t$. We see that 
$$
r(v,t) = \left( {|S(v,t)|}/{c_m} \right)^{1/m} = \left(\frac{A_1 \deg^-(v,t) + A_2}{c_m t}\right)^{1/m},
$$
where $c_m$ is the volume of the unit ball; for example, in 2-dimensions with the Euclidean metric, $c_2 = \pi$. 

As typical in random graph theory, we shall consider only asymptotic properties of $G_n$ as $n\rightarrow \infty$. We say that an event in a probability space holds \emph{asymptotically almost surely} (\aas) if its probability tends to one as $n$ goes to infinity. We emphasize that the notations $o(\cdot)$ and $O(\cdot)$ refer to functions of $n$, not necessarily positive, whose growth is bounded. Since we aim for results that hold a.a.s., we will always assume that $n$ is large enough. 

\section{Graph properties of the SPA model}\label{sec:graphprop}

In this section we investigate typical properties of  graphs  produced by the inhomogeneous SPA model, aiming to  use the results to infer the spatial distances between the nodes. A central observation is that in the inhomogeneous SPA model with a locally constant density function, the probability of an edge forming from a new node $v_t$ to an existing node $v$ at time $t$ equals
\[
\pp \Big( (v_t,v) \in E(G_n) \Big) =p\int_{S(v,t)} \rho d\mu =p \sum_\ell \rho_\ell \,|S(v,t) \cap \reg_\ell|.
\]

In the analysis of the original SPA model from \cite{spa1}, we find that  spheres of influence of nodes that are born early typically shrink rapidly, while nodes born late start with small spheres of influence. A node would have to be quite close to the boundary of its region with another one for the effect of any other region to be felt. With this assumption, the expression for the link probability is very similar to that of the link probability of the original SPA model. Therefore, it seems reasonable to expect that the graph formed by nodes in a region $\reg_{\ell}$ with  local density $\rho_{\ell}$ behaves like an independent SPA model of density $\rho_{\ell}$. Our results will show that this expectation is justified and can be made rigorous.

To be specific, assume that nodes in the SPA model do not arrive at fixed, discrete, time instances $t$, but instead arrive according to a homogeneous Poisson process with rate 1. (This will not significantly change the analysis but is a convenient assumption.) Then, the process inside a region $\reg$ with density $\rho$ will behave like the SPA model with the same parameters $A_1$, $A_2$ and $p$, but with points arriving according to a Poisson process with rate $\rho$. This means that in each time interval we expect $\rho$ points to arrive, and the expected time interval between arrivals equals $1/\rho$.  If we use $v_t$ to denote the $t$-th node arriving, then the arrival time $a(t)$ of $v_t$ is approximately $t/\rho$, and thus the volume of the sphere of influence of an existing node $v$ at the time that $v_t$ is born equals 
\[
|S(v,a(t))|=\frac{A_1\deg^-(v,a(t))+A_2}{a(t)}\approx \frac{\rho A_1\deg^-(v,a(t))+\rho A_2}{t}.
\]
Thus, in the analysis of the degree of an individual node, we expect a node $v$ in the inhomogeneous SPA model to behave like a node in the original SPA model with parameters $\rho(v)A_1$, $\rho(v)A_2$ instead of $A_1$, $A_2$, where the degree of node $v$ at time $t$ in the inhomogeneous SPA model corresponds to the degree of a node at time $a(t)$ in the corresponding SPA model. The following theorems show that this is indeed the case.

\begin{theorem}\label{degconc}
Let $\omega=\omega(n)$ be any function tending to infinity together with $n$, and let $\eps > 0$. The following holds with probability $1-o(n^{-1})$. For every node $v$ for which 
$$
\deg^-(v,n)=k=k(n) \geq \omega^2 \log n
$$
and for which 
\begin{equation}\label{eqn:deltabound}
\delta(v) \geq (1 + \eps)\left(\frac{A_1 k + A_2}{c_m n}\right)^{1/m}=(1+\eps)r(v,n),
\end{equation} 
it holds that for all values of $t$ such that $\max\{t_v,T_v\} \le t \le n$, 
\begin{equation}\label{eqn:deg}
\deg^-( v,t) = (1+o(1)) k \left(\frac{t}{n}\right)^{p\rho (v)  A_1}.
\end{equation}
Times $T_v$ and $t_v$ are defined as follows: 
\begin{equation}\label{eqn:tv}
T_v=n \left(\frac{\omega \log n}{k}\right)^{\frac{1}{p \rho (v) A_1}},\quad
t_v= (1+\eps)\left(\frac{A_1k}{\delta(v)^m c_m n^{p\rho (v) A_1}}\right)^{\frac{1}{1-p\rho(v) A_1}}.
\end{equation}
\end{theorem}

Condition~(\ref{eqn:deltabound}) on $\delta(v)$ ensures that at time $n$, $S(v,n)$ is completely contained in $\reg(v)$ (deterministically). In fact, due to the additional multiplicative factor of $(1+\eps)$, $S(v,n)$ is some distance removed from the boundary of $\reg(v)$. The expression for $T_v$ is chosen so that at this time node $v$ has \aas~at least $\omega \log n$ neighbours. Likewise, $t_v$ is chosen such that at this time \aas~the sphere of influence has shrunk so that its radius is sufficiently smaller than $\delta (v)$, again with some extra room to spare. The implication of this theorem is that once a node accumulates at least $\omega \log n$ neighbours and its sphere of influence has shrunk so that it does not intersect neighbouring regions, its behaviour can be predicted with high probability until the end of the process, and is completely governed by its region, and no others. In particular, it follows that from time $\max\{t_v,T_v\}$ onwards the sphere of influence is completely contained in $\reg(v)$.  

For most vertices, the moment when they first achieve $\omega\log n$ neighbours ($T_v$,) will come before the moment that their sphere of influence has shrunk so that it is well contained in the region ($t_v$). Indeed, consider a vertex $v$ of degree at least $\omega \log n$ for which this is not the case. Let $T$ be the moment when the vertex reaches in-degree $\omega \log n$. By definition, the sphere of influence of $v$ at this time $T$ has a radius of influence of order $\left(\frac{\omega\log n}{T}\right)^{1/m} $. If $ T\gg \omega\log n$, then the radius is $o(1)$, and the probability that $v$ is this close to the border is also $o(1)$. The only vertices for which potentially the radius at time $T$ could be fairly large are those vertices for which  $T=O(\omega\log n) $. Thus, these are the oldest vertices. These vertices do have high degree, but their spheres of influence still tend to shrink over time, so most of their edges will be acquired after time $t_v$, that is, when their sphere of influence has shrunk to be contained in the region. 

We can use the results on the degree to show that each graph induced by one of the regions $\reg_{\ell}$  has a power law degree distribution.
Let $N_{\ell}(j,n)$ denote the number of nodes of degree $j$ at time $n$ in the region $\reg_{\ell}$. The proof of the following result is a straightforward adaptation of the differential equations method used to prove the counterpart result for the uniform model (see~\cite{spa1}). Since this theorem is not needed to prove the main result of this paper, the proof is omitted here.

\begin{theorem}\label{thm:powerlaw}
A.a.s.\ the graph induced by the nodes in region  $\reg_{\ell}$ has  a power law degree distribution with coefficient $1 + 1/(p\rho_{\ell}A_1)$. Precisely, \aas\ for any $1 \le \ell\le k^m$ there exists a constant $c_{\ell }$ such that for any $1 \ll j \le j_f= \left( n/\log^8 n \right)^{\frac{p \rho_{\max} A_1}{4p\rho_{\max} A_1+2}}$,
\begin{equation*}
N_{\ell}(j,n) =  (1+o(1))c_{\ell} j^{-(1 + \frac{1}{p\rho_\ell A_1})} q_{\ell} n.
\end{equation*}
Moreover,  \aas\ the entire graph generated by the inhomogeneous SPA model has a degree distribution whose tail follows a power law with coefficient  $1 + 1/(p\rho_{\max} A_1)$.
\end{theorem}

The number of edges also validates our hypothesis that a region of a certain density behaves almost as a uniform SPA model with adjusted parameters. In the original SPA model with parameters $A_1$ and $A_2$ replaced by $\rho A_1$, $\rho A_2$ and $p$, the average out-degree is approximately $\frac{p\rho A_2}{1-p \rho A_1}$, as per~\cite[Theorem 1.3]{spa1}. The following theorem shows that the subgraph induced by one of the regions has the equivalent expected number of edges. This theorem also shows that \aas\ the number of edges that cross the boundary of a region is of smaller order than the number of edges completely contained in that region. Thus, almost all edges have both endpoints in the same region. 

\begin{theorem}\label{thm:noEdges}
A.a.s., for all regions $\reg_{\ell}$ of density $\rho_{\ell}$,  $|V(G_n)\cap \reg_\ell |=(1+o(1))q_{\ell}n$.
Moreover,
$$
\E ( \{ (u,v)\in E(G_n)\, |\, u,v\in\reg_\ell\}|) = (1+o(1))\frac{p\rho_{\ell}A_2}{1-p\rho_{\ell}A_1}q_{\ell}n.
$$
Furthermore, \aas
$$|\{(u,v) \in E(G_n) :\reg(u) \neq \reg(v)\}| = o(n).$$
\end{theorem}

Here we see that we need the condition $p \rho_{\max} A_1 < 1$. If  $p \rho_{\max} A_1 \geq 1$, then the number of edges would grow superlinearly.

Our ultimate goal is to derive the pairwise distances between the nodes in the metric space through an analysis of the graph. The following theorem, obtained using the approach of~\cite{geo1}, provides an important tool. Namely, it links the number of common in-neighbours of a pair of nodes to their (metric) distance. Using this theorem, we can then infer the distance from the number of common in-neighbours. 

The theorem distinguishes three cases. If $u$ and $v$ are relatively far from each other, then they will have no common neighbours. If the nodes are very close, then the number of common neighbours is approximately equal to a fraction $p$ of the degree of the node of smallest degree. The third case provides a `sweet spot' where the number of common neighbours is a direct function of the metric distance and the degrees of the nodes. For any two nodes $u$ and $v$, let $cn(u,v,t)$ denote the number of common in-neighbours of $u$ and $v$ at time $t$.

\begin{theorem}\label{thm:cn2}
Let $\omega=\omega(n)$ be any function tending to infinity together with $n$, and let $\eps > 0$. The following holds \aas\
Let $u$ and $v$ be nodes of final degrees $\deg(u,n)=k$ and $\deg(v,n)=j$ such that $\reg=\reg(u)=\reg (v)$, 
and $k \geq j \ge \omega^2 \log n$.
Let $\rho=\rho(v)=\rho(u)$ and let $T_v = n\left(\frac{\omega \log n}{j}\right)^{\frac{1}{p\rho A_1}}$, and assume that 
\[
\delta(v)^m \geq cj\mbox{ and }\delta (u)^m\geq ck, \ \ \ \ \mbox{ where } 
c=(1+\eps) \left(\frac{A_1}{c_m n^{p\rho A_1}T_{v}^{1-p\rho A_1}}\right).
\]
Let $d(u,v)$ be the distance between $u$ and $v$ in the metric space. Then, we have the following result about the number of common in-neighbours of $u$ and $v$:
\begin{enumerate}
\item[Case 1.] If
$$d(u,v) \ge \eps \left(\frac{(\omega \log n) (k/j)}{T_{v}}\right)^{1/m}$$ 
then 
$
cn(u,v,n)=O(\omega \log n).
$

\smallskip
\item[Case 2.] If 
$$
d(u,v) \le \left( \frac {A_1 k+A_2}{c_m n} \right)^{1/m} - \left( \frac {A_1 j + A_2}{c_m n} \right)^{1/m}
$$
then  $cn(u,v,n) = (1+o(1))p j.$

\smallskip
\item[Case 3.] If 
$$
\left( \frac {A_1 k+A_2}{c_m n} \right)^{1/m} - \left( \frac {A_1 j + A_2}{c_m n} \right)^{1/m} < d(u,v) < \eps \left(\frac{(\omega \log n) (k/j)}{T_{v}}\right)^{1/m},
$$
then 
\begin{equation}
\label{eqn:cn}
cn(u,v,n) = C j n^{-\alpha}k^\alpha  d^{-m\alpha}  \left( 1+o(1)+O \left( \left( \frac{j}{k} \right)^{1/m} \right) \right),
\end{equation}
where 
\[
\alpha=\frac{p\rho A_1}{1-p\rho A_1} \ \ \ \ \ \mbox{ and }\ \ \ \ \ C=pA_1^\alpha c_m^{-\alpha}.
\]
Note that, if $j\ll k$, then we have a precise asymptotic formula for $cn(u,v,n) $. If $j$ and $k$ are approximately equal, then the formula only states that $cn(u,v,n)=\Theta(j n^{-\alpha}k^\alpha  d^{-m\alpha} )$.
\end{enumerate}
\end{theorem}

\section{Reconstruction of Geometry}\label{sec:Sim}

We set out to discover the character of nodes in a network purely through link structure, and to quantify the similarities. Spatial models allow us a convenient definition of similarity: distances between nodes. In examining the SPA model, the number of common neighbours allows us to uncover a good approximation of pairwise distances, a first step in the reconstruction of the geometry. 

\paragraph{Description of Model Used:}

For simulations, we use an inhomogeneous SPA model that we call a \emph{diagonal layout}, which has 4 `clusters' of identical high density, with $m=2$. In the diagonal layout, $k=4$ and the 4 regions $(x,x)$, $1 \leq x \leq 4$, are dense, with the others sparse. We will use `dense region' and `sparse region' to denote the union of all regions with densities $\rho_d$ and $\rho_s$, respectively. For ease of notation, we note that $\frac{1}{16}(4\rho_d +12\rho_s)=1$, so $\rho_s = 4/3-\rho_d/3$. Thus it is enough to provide the value of $\rho_d$ only. In Figure~\ref{fig:diagmod} we see an example of the diagonal layout with nodes and edges, and we also see evidence that the densest region does dominate the power law degree distribution. The yellow line is the prediction for the degree distribution with the power law exponent based on the maximum density, as in Theorem~\ref{thm:powerlaw}.

\begin{figure}[h]
\begin{center}
\includegraphics[scale=0.4]{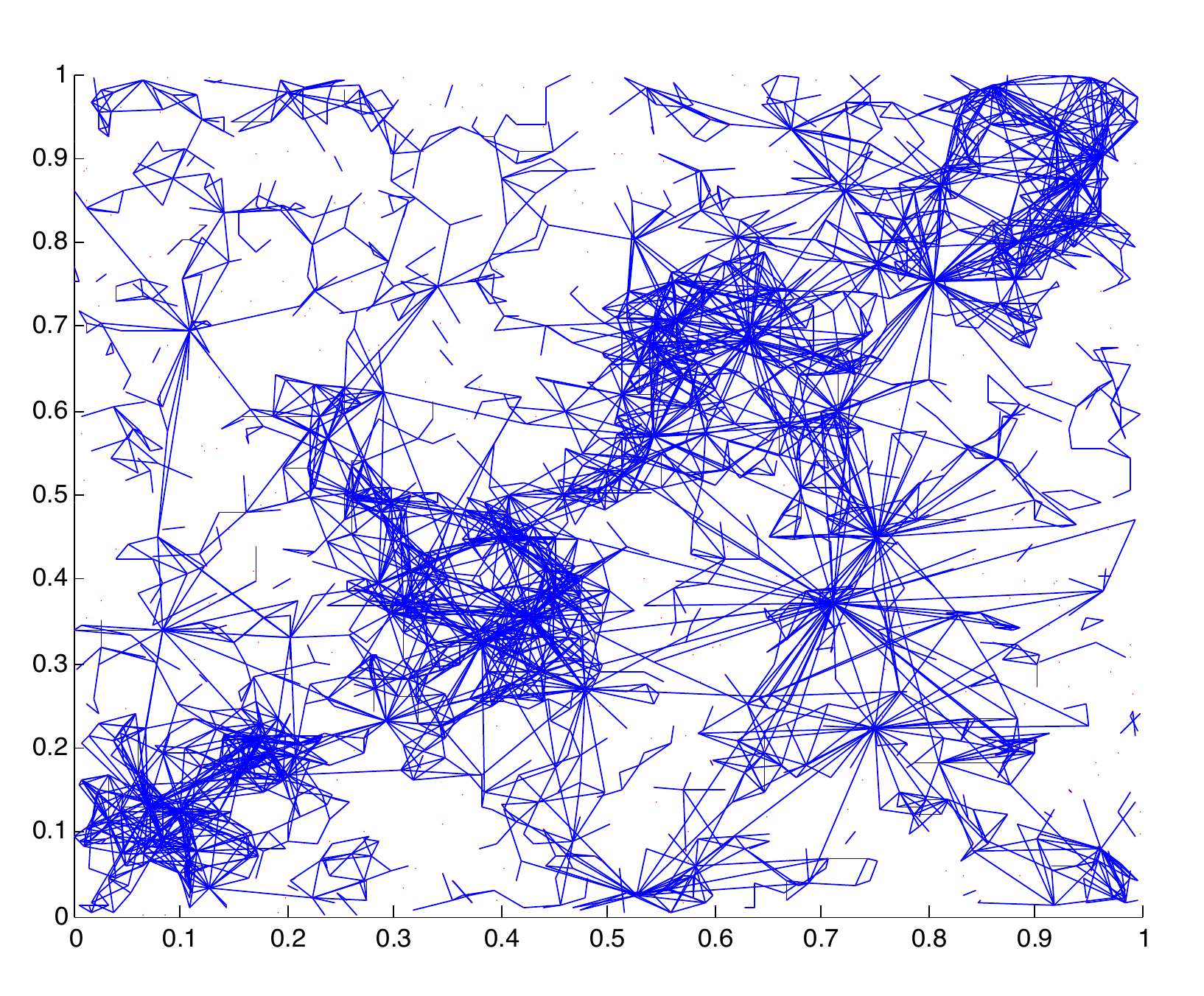}
\includegraphics[scale=0.4]{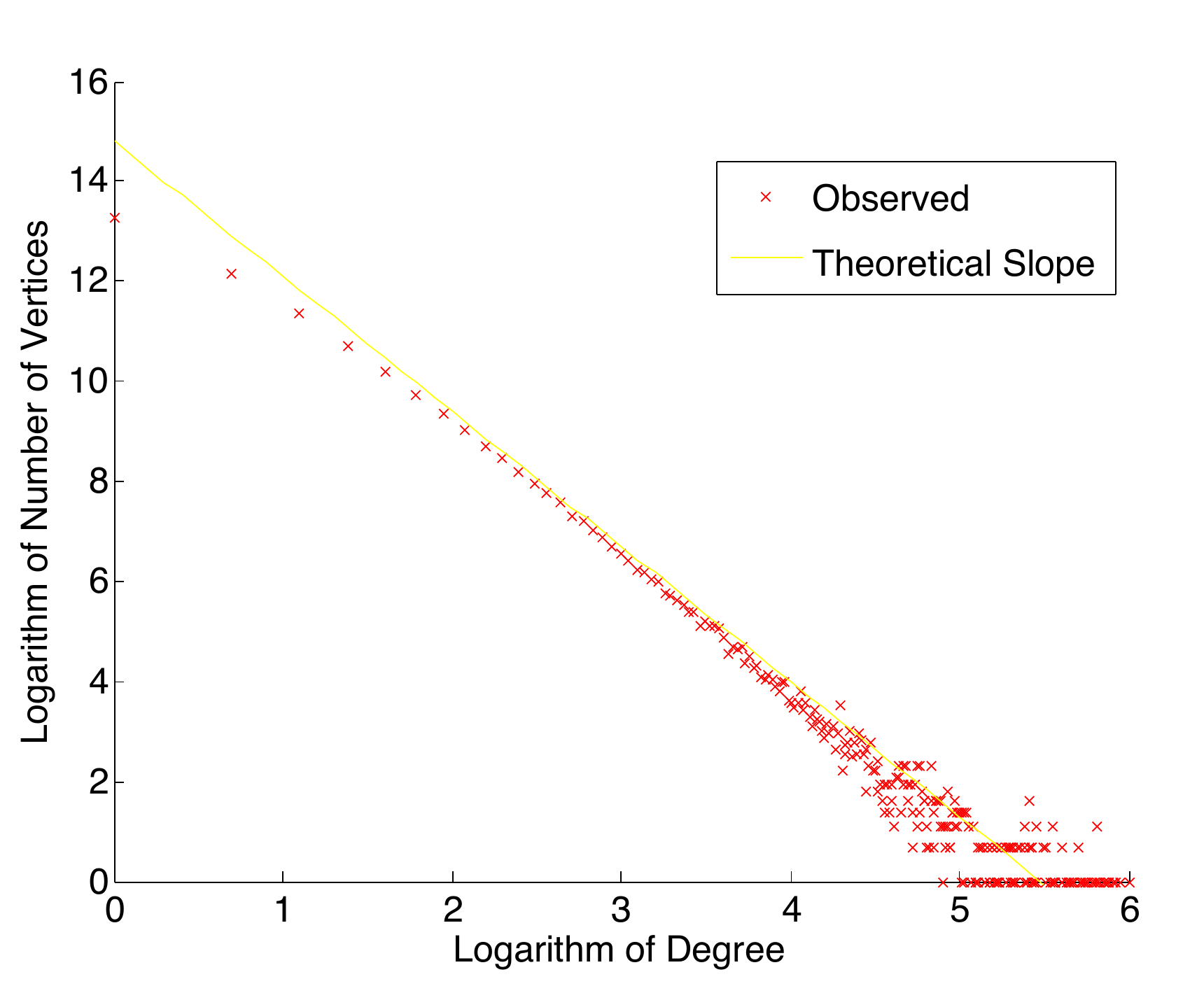}
\end{center}
\caption{Left: diagonal layout, $n=1,000$, $p=0.6$, $\rho_d=1.6$, $A_1 = 0.7$, $A_2 = 2.0$; Right: degree distribution $n=1,000,000$, $p=0.7$, $\rho_d=1.2$, $A_1 = 0.7$, $A_2 = 1.0$.}
\label{fig:diagmod}
\end{figure}

Our estimator for the distance is derived from Case 3 of Theorem~\ref{thm:cn2}, and in particular Equation~(\ref{eqn:cn}), ignoring the error term. This leads to the following formula for the estimated distance $\hat{d}$. For a pair of nodes $u,v$ with $\deg^-(u,n)=k$ and $\deg^-(v,n)=j$,  $k\geq j$, whose distance is such that Case 3 applies, this estimate is given by: 
\begin{equation}\label{eq:DistCN}
\hat{d}(u,v)=\left( \frac{p^{\gamma} A_1 j^\gamma k}{nc_m cn(u,v,n)^\gamma}\right)^{\frac{1}{m}},
\end{equation}
where $\gamma =\frac{1-p\rho (u) A_1}{p\rho( u) A_1}$ (note that $\gamma=\frac{1}{\alpha}$ with $\alpha $ as in Equation (\ref{eqn:cn}).) If the density is uniform, that is, if $\rho(u)=1$ for all $u$, then the above estimator is the same as our original estimator: Equation~(7) from~\cite{geo1}. 

Since a relationship between the spatial distance only exists in Case 3 of Theorem~\ref{thm:cn2}, we try to eliminate pairs to which one of the other cases applies. Pairs which are in Case 1 are very close, and for such pairs, the expected number of common neighbours is $p \deg (v,n) = pj$. In an attempt to avoid this case, we filter out all pairs where the number of common neighbours is greater than $pj/2$. Pairs that are in Case 2 are so far apart that their spheres of influence have overlap for a very short time, if at all. We try to avoid this case by eliminating pairs with 10 or fewer common neighbours. 

To see the effect of the non-uniform density, we first apply the original estimator to our diagonal layout. In other words, we define the estimated distance as in Equation~(\ref{eq:DistCN}), but taking $\rho(u)=1$ for all $u$, and we are applying this estimator to the points obtained from a non-uniform distribution.  The motivation of this experiment is that, when applying our techniques to real-life data, we are not likely to know the local density of a node. Figure~\ref{diag:Est} (left side) gives the estimated versus real distance for a graph with $n=100,000$ nodes, generated via the SPA model from the diagonal layout with parameters $p=0.7$, $\rho_d=1.6$, $A_1 = 0.7$, $A_2 = 2.0$. After filtering as described above, 2,270 pairs are left. 

The figure shows that the approach of assuming uniform density leads to a consistent overestimate of the distance for the nodes. This may seem counterintuitive. The trouble lies with the estimator's assumption about a node's age, which is based on its final in-degree. A node in $\reg_d$ has more neighbours than is expected when one assumes uniform density, and thus the node is thought to be much older than it actually is. This confounds the distance estimator.

\begin{figure}[h]
\begin{center}
\includegraphics[width=0.4\textwidth]{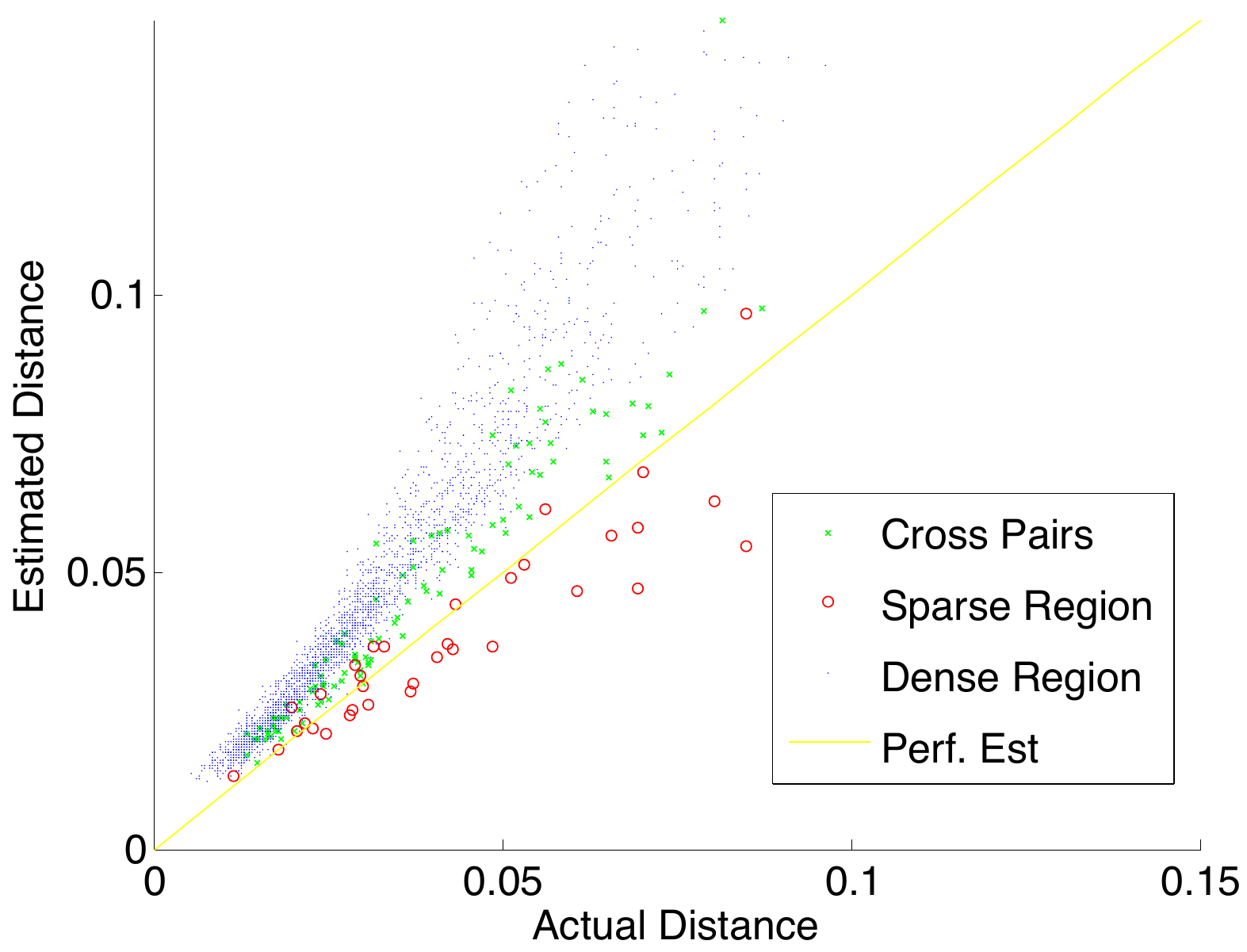}
\includegraphics[width=0.4\textwidth]{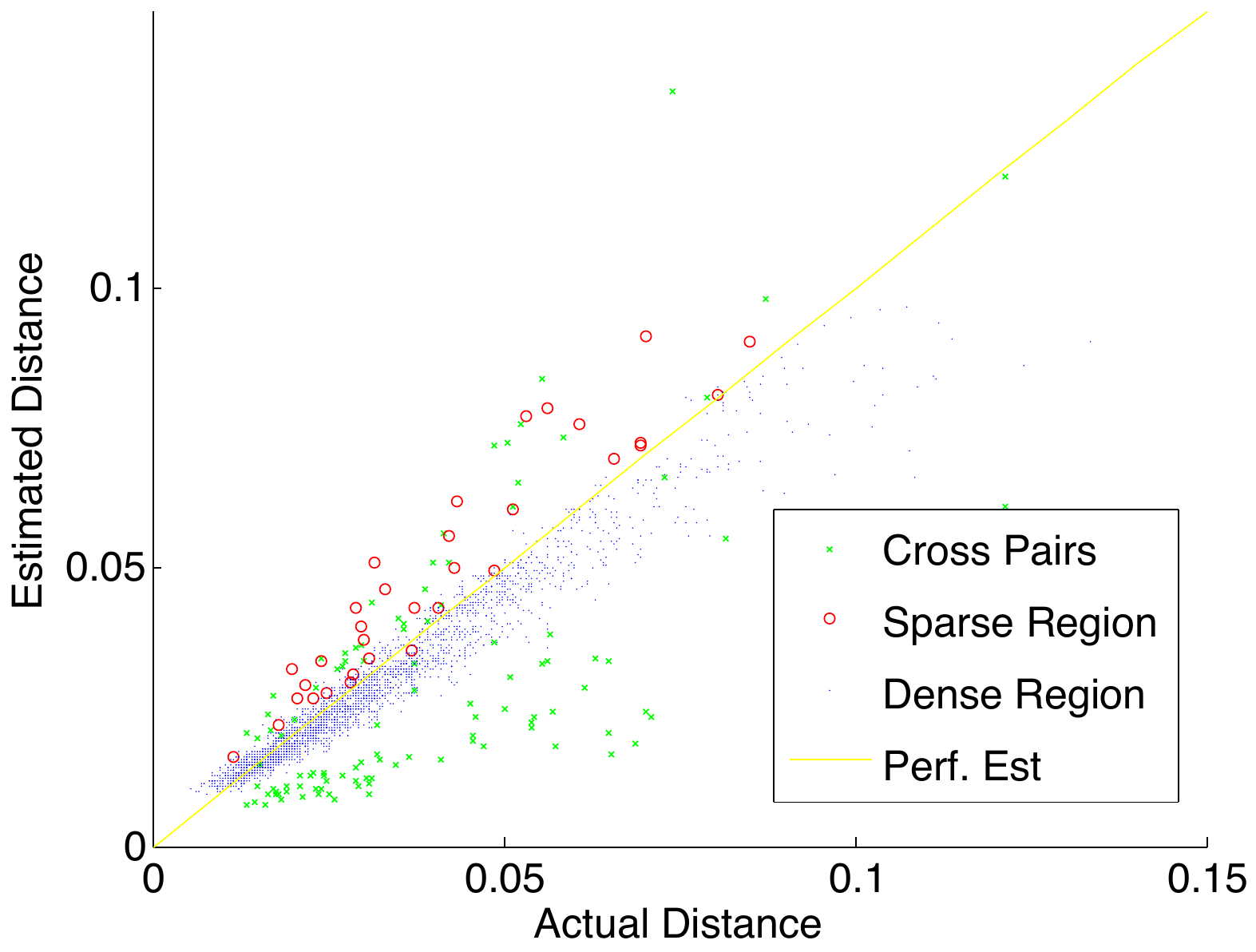}
\end{center}
\caption{SPA model, $n=100,000$, diagonal layout, $p=0.7$, $\rho_d=1.6$, $A_1 = 0.7$, $A_2 = 2.0$, actual vs.\ estimated distances for pairs of nodes; Left: using original estimator; Right: using new estimator, density known}
\label{diag:Est}
\end{figure}

Using the same simulation results, we now apply the estimator from Equation~(\ref{eq:DistCN}), and use our knowledge about $\rho (u)$.  The figure on the right in Figure~\ref{diag:Est} shows the estimated distance using Equation~(\ref{eq:DistCN}) vs.\ actual node distance. The results indicate that our new estimator is significantly more accurate in predicting distances for the pairs of nodes in the dense region. 

Let us mention that the estimation for pairs in the sparse region is still not accurate, while the estimation for cross-border pairs appears to be even worse. This is likely caused by the fact that nodes that are involved in cross-border pairs, and in sparse region pairs, and that have enough common neighbours to qualify to be included, are likely the older nodes, i.e.~they are born near the beginning of the process. For such nodes, in the early stages there is likely some overlap between their sphere of influence and the bordering, dense regions. Thus, the degree likely does not follow the prediction from Theorem~\ref{degconc}, which, in turn, affects the performance of the distance estimator for those pairs.

Better performance for cross-border pairs could possibly be obtained by using a linear combination of the densities in Equation~(\ref{eq:DistCN}). However, we will see in what follows that better performance for all pairs occurs when we use the data itself to estimate the density. Also, we point out that the pairs in the dense region constitute the large majority of all pairs. Moreover, the dense regions are those that are most likely to correspond to communities of interest. Therefore, accurate prediction for pairs in these regions is most important.

\paragraph{Estimating the density:}

In real-world situations, we cannot assume to know the density of the region containing a given node. In fact, the density of the region containing a node is an important part of the `second layer of meaning' which we aim to extract from the graph. Here we will show that our theoretical results give us a tool for estimating the local density around a node, using only its neighbourhood. We also apply our distance estimator once more, this time using the estimated density for our formula.

Using the theoretical results obtained from the previous section, we can estimate the density of the region $\reg(v)$ containing a given node $v$ from the average out-degree of the in-neighbours of  $v$. As per Theorem~\ref{thm:noEdges}, the average out-degree in $\reg_\ell $ is approximately 
\begin{equation}\label{eqn:outdeg}
\frac{p\rho_{\ell}A_2}{1-p\rho_{\ell}A_1}.
\end{equation}
If we have a large enough set of nodes from the same region, then we can use the formula above to estimate the density of the region. Consider a node $v$, and make two assumptions: ($i$) almost all neighbours of $v$ are contained in $\reg (v)$, and ($ii$) the neighbours of $v$ form a representative sample of all nodes of $\reg (v)$. Simulations show that these assumptions are justified and allow us to make an estimate for $\rho(v)$. Assumption ($i$) is additionally justified by the second part of Theorem~\ref{thm:noEdges}, which states that the number of edges crossing the border is negligible compared to the total number of edges.

Set $\overline{\deg}^+(v)$ to be the average out-degree of the in-neighbours of $v$. Specifically,
\[
\overline{\deg}^+(v)=\frac{1}{\deg^-(v)}\sum_{u\in N^-(v)} \deg^+(u).
\]
Given our assumptions, an estimator for the density in $\reg (v)$, denoted by $\hat{\rho}(v)$, can be derived from this average out-degree, using Equation~(\ref{eqn:outdeg}):
\[
\hat{\rho}(v)=\frac{\overline{\deg}^+(v)}{pA_2 + pA_1\overline{\deg}^+(v)},
\]
where $N^-(v)$ is the set of in-neighbours of $v$.

The left side of Figure~\ref{fig:histAvgOD} shows a histogram of the values of $\overline{\deg}^+(v)$ for our simulated graph. Displayed are the results for nodes with $\deg^-(v) \geq 10$. The graph is obtained from the SPA model where points have the previously described diagonal layout, with density $\rho_d =1.6$ in the dense region, and consequently density $\rho_s = 0.8$ in the sparse region. For these parameters, Equation~(\ref{eqn:outdeg}) gives a theoretical value of 5.85 for $\overline{\deg}^+(v)$ if node $v$ lies in the dense region, and a value of 1.45 if $v$ lies in the sparse region. 

We see in Figure~\ref{fig:histAvgOD} (left side) that the values of $\overline{\deg}^+(v)$ in the dense region are quite accurate, with peaks occurring around the calculated value of 5.85. For   the sparse region, the peaks occur around 2.5, giving an estimate for the average out-degree which is higher than expected. Likely, this is caused by nodes in the sparse region that are located close to the border, and thus are likely to have neighbours in the dense region. Such nodes also tend to have high degree, and our condition on the minimum degree favours the `rich' sparse region nodes.

Figure~\ref{fig:histAvgOD} (right side) gives a histogram of the estimated densities of the nodes. For nodes in the dense region, the true value is 1.6, and we see a good estimation of this value for these nodes. For nodes in the sparse nodes, the true value is 0.8, while the peak of the estimated densities occurs around 1.15, and almost all values are greater than 0.8. Again, this is likely caused by nodes whose sphere of influence overlapped with the dense region. 

To obtain better performance for nodes in the sparse region, we propose to base our estimated density for a node $v$ in the sparse region only on the out-degree of neighbours of $v$ of low in-degree; such neighbours are young and so the sphere of influence of $v$ had shrunk, and thus was more likely to be fully inside the sparse region, when the neighbours were born. To obtain density estimates for nodes with small in-degree, we can take the second neighbourhood to compute the average out-degree. Nodes with small in-degrees are young, so even second neighbours are likely to be close. We plan to explore these possibilities in future work.
 
\begin{figure}[h]
\begin{center}
\begin{tabular}{cc}
\includegraphics[scale=0.4]{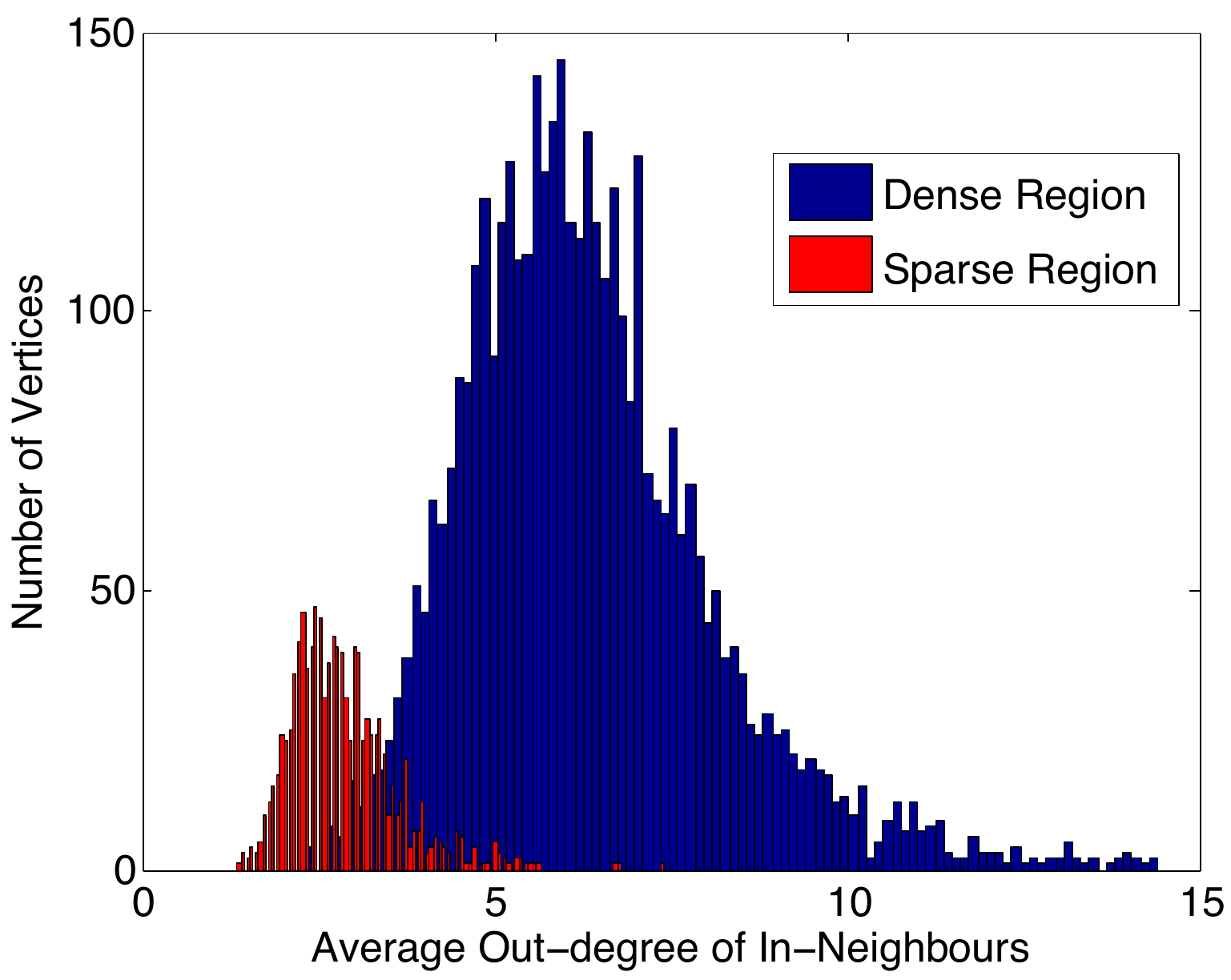}&
\includegraphics[scale=0.4]{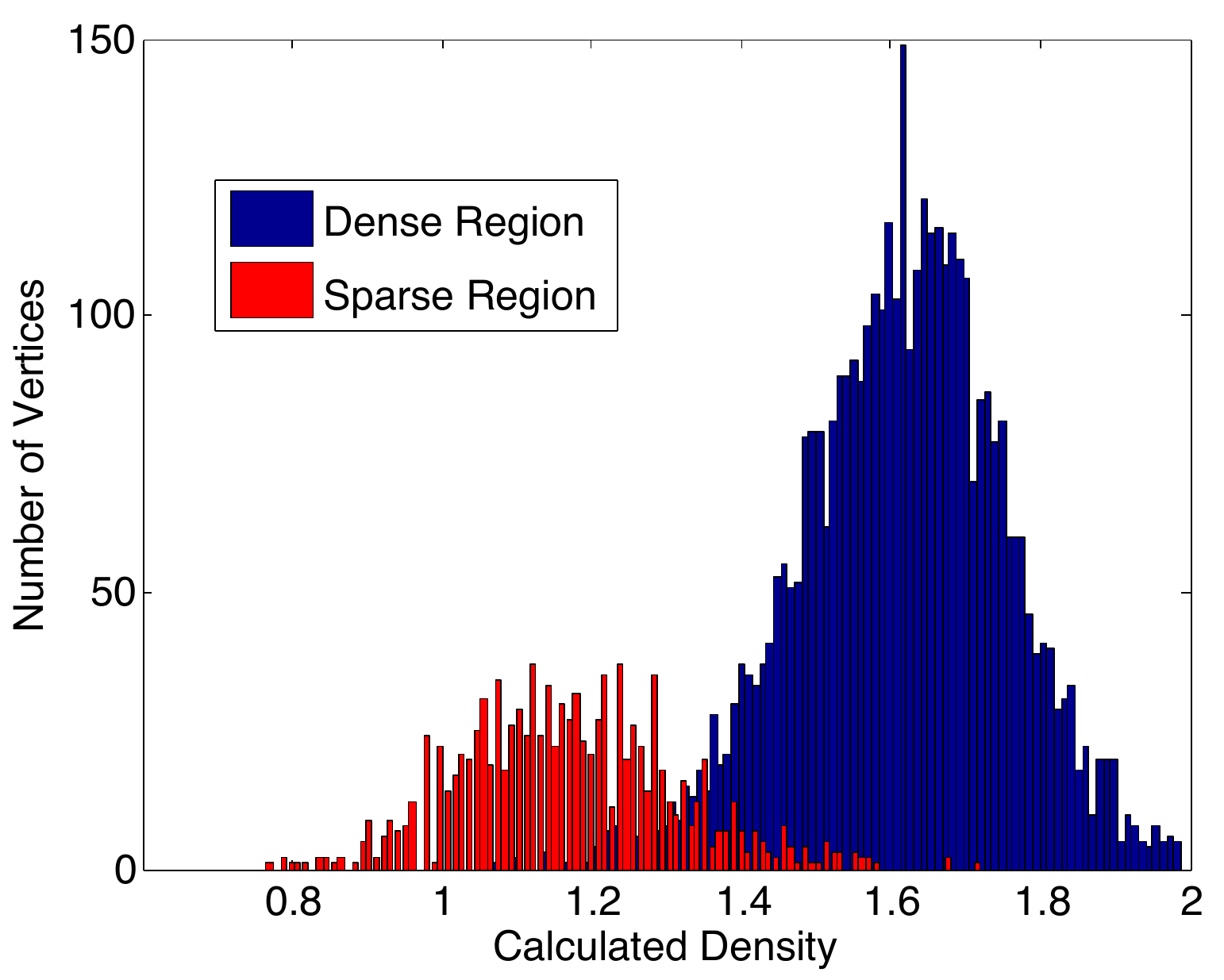}
\end{tabular}
\end{center}
\caption{Diagonal layout, $n = 100,000$, $\rho_d=1.6$, $A_1 = 0.7$, $A_2 = 2.0$, $p=0.6$; Left: average out-degree of the in-neighbours; Right: calculated density from average out-degree. }
\label{fig:histAvgOD}
\end{figure}

Finally, we use $\hat{\rho}$, and known values of all other parameters, to calculate the distance between the nodes based on the number of common neighbours, Equation~(\ref{eq:DistCN}), using the same simulation results as those we used earlier. Here we use the calculated density of the node of higher degree in the distance formula. (Using the lower degree node gives similar results.) The results are seen in Figure~\ref{fig:EstDensAndDist}. 

The figure shows that there is very good agreement between calculated and estimated densities. In fact, we see that the agreement is greatly improved for the cross-border pairs, and also better for the sparse region pairs. This can be understood as follows. The distance estimator is derived indirectly from Theorem~\ref{thm:exp_degree}, which predicts the approximate degree of a node throughout the process, based on its final degree and the density of its region. For nodes in the sparse region which have a sizeable number of neighbours in the sparse region, the degree will be larger than predicted using this method but also, the density estimator will predict a higher density. So the estimated density is a better indicator of the behaviour of the degree than the real density, and thus the distance estimator gives better performance. This indicates that this last variation of the distance estimator is the most robust against local fluctuations in density. Thus we have a good prognosis for the applicability of the estimator on real data, where such fluctuations are to be expected.

\begin{figure}[h]
\begin{center}
\includegraphics[scale=0.45]{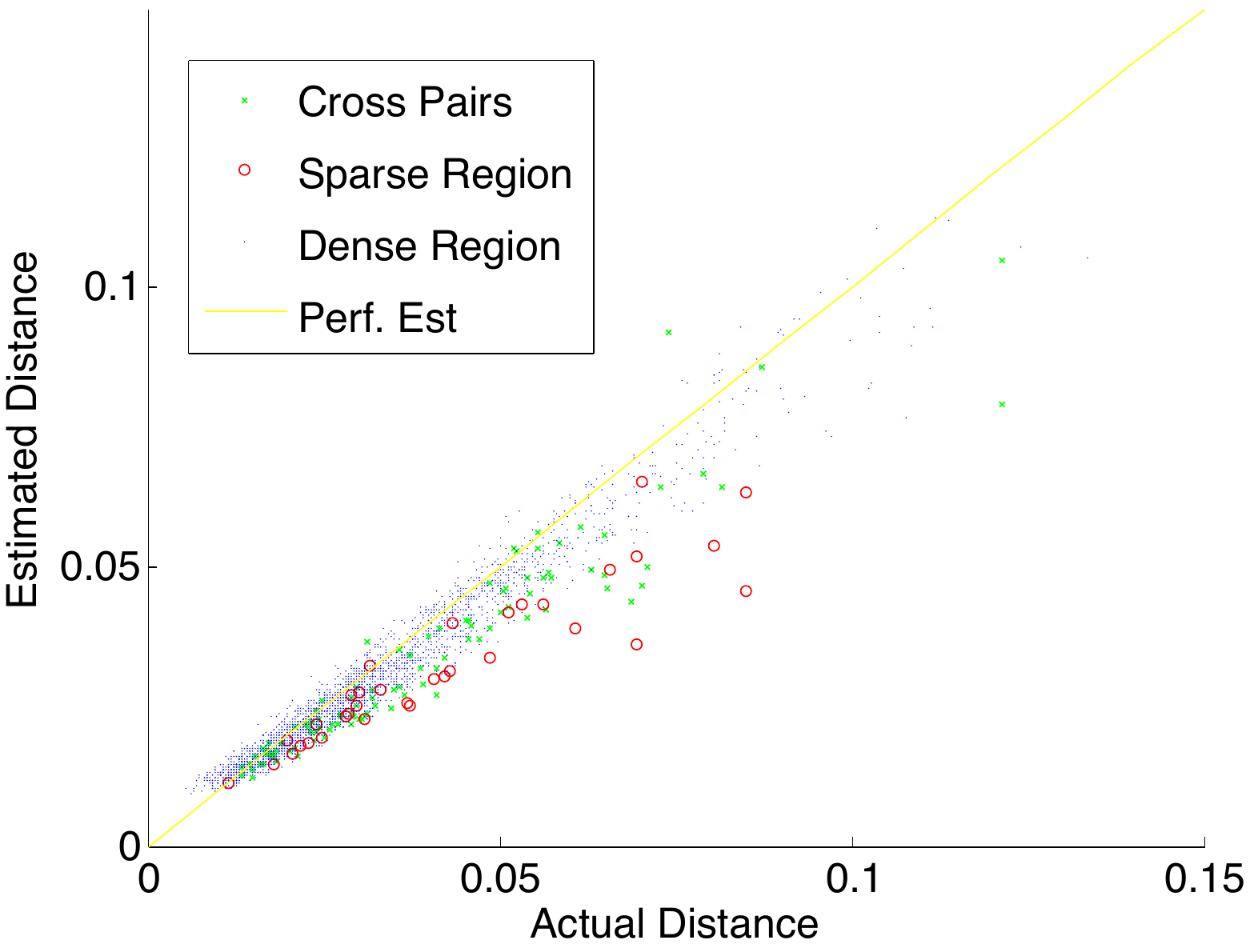}
\end{center}
\caption{Diagonal layout, $n = 100,000$, $\rho_d=1.6$, $A_1 = 0.7$, $A_2 = 2.0$, $p=0.7$:
Distance estimation using estimated density from the node of greater final degree, all other parameters known}
\label{fig:EstDensAndDist}
\end{figure}

\section{Proofs}\label{sec:proofs}

In this section, we give the proofs of the main theorems. Our results all refer to typical behaviour of the random SPA model process, and are asymptotic in $n$, the number of vertices. We will sometimes use the stronger notion of \wep\ in favour of the more commonly used \aas, since it simplifies some of our proofs. We say that an event holds \emph{with extreme probability} (\wep), if it holds with probability at least $1-\exp (-\omega(n) \log n)$ as $n \to \infty$, where $\omega(n)$ is any function tending to infinity together with $n$.  Thus, if we consider a polynomial number of events that each holds \wep, then \wep\ (and hence also \aas) all events hold.

\bigskip

First we state and prove a theorem that bounds the in-degree of any node, regardless of its distance of the boundary. 

\begin{theorem}\label{thm:exp_degree} 
Let $\omega=\omega(t)$ be any function tending to infinity together with $t$. The expected in-degree at time $t$ of a node $v_i$ born at time $i \ge \omega$ satisfies
\begin{eqnarray*}
\E(\deg^-(v_i,t)) &\leq& (1+o(1))\frac{A_2}{A_1} \left(\frac{t}{i}\right)^{p \rho_{\max}A_1} - \frac {A_2}{A_1}. \\
\E(\deg^-(v_i,t)) &\geq& (1+o(1))\frac{A_2}{A_1} \left(\frac{t}{i}\right)^{p (2^{-m} \rho(v)+(1-2^{-m})\rho_{\min})A_1} - \frac {A_2}{A_1}.
\end{eqnarray*}
Moreover, for any node $v_i$ born at time $i \ge 1$ we have
\begin{equation*}
\E(\deg^-(v_i,t)) \leq \frac{e A_2}{A_1} \left(\frac{t}{i}\right)^{p \rho_{\max}A_1} - \frac {A_2}{A_1}.
\end{equation*}
\end{theorem}

\begin{proof}
In order to simplify calculations, we make the following substitution:
\begin{equation*}
Y(v_i,t) = \deg^-(v_i,t) + \frac{A_2}{A_1}=\frac{t}{A_1}|S(v_i,t)|.
\end{equation*}
It follows immediately from the definition of the process that $Y(v_i,i)=A_2/A_1$ and for $t \ge i$
$$
Y(v_i,t+1) =
\begin{cases}
      Y(v_i,t)+1, & \text{with probability at most }p \rho_{\max}\frac{A_1}{t} Y(v_i,t),\\
      Y(v_i,t), & \text{otherwise}.
   \end{cases}
$$

We couple $Y(v_i,t)$ with another random variable $X(v_i,t)$ so that $Y(v_i,t) \le X(v_i,t)$ for $t \ge i$. Random variable $X(v_i,t)$ is defined as follows: $X(v_i,i)=A_2/A_1$ and for $t \ge i$
$$
X(v_i,t+1) =
\begin{cases}
      X(v_i,t)+1, & \text{with probability }p \rho_{\max}\frac{A_1}{t} X(v_i,t)\\
      X(v_i,t), & \text{otherwise}.
   \end{cases}
$$
Finding the conditional expectation,
\begin{equation*}
\mathbb{E}(X(v_i,t+1)~~|~~X(v_i,t)) =X(v_i,t)\left(1+\frac{p\rho_{\max}A_1}{t}\right).
\end{equation*}
Taking expectations, we get
$$
\mathbb{E}(X(v_i,t+1)) = \mathbb{E}(X(v_i,t)) \left( 1+\frac{p\rho_{\max}A_1}{t} \right),
$$
and, since $X(v_i,i) = {A_2}/{A_1}$,
\begin{align*}
\mathbb{E}(X(v_i,t)) &= \frac{A_2}{A_1}\prod_{j=i}^{t-1} \left( 1+\frac{p\rho_{\max}A_1}{j} \right)&\\
&\le  \frac{A_2}{A_1} \exp\left(\sum_{j=i}^{t-1} \frac{p\rho_{\max}A_1}{j}\right)&\\
&\le \frac{A_2}{A_1} \exp\left(p\rho_{\max}A_1 \left(\log \left( \frac {t}{i} \right)+1/i\right) \right)  \\
&\le \frac{eA_2}{A_1} \left(\frac{t}{i}\right)^{p\rho_{\max}A_1}.
\end{align*}
If $i \ge \omega$, we have
$$
\mathbb{E}(X(v_i,t)) = (1+o(1)) \frac{A_2}{A_1} \exp\left(\sum_{j=i}^{t-1}\frac{p\rho_{\max}A_1}{j}\right) =(1+o(1))\frac{A_2}{A_1} \left(\frac{t}{i}\right)^{p\rho_{\max}A_1}.
$$
This shows the upper bound.

For the lower bound, we first observe that for all nodes $v$, $|S(v,t)\cap \reg (v)|\geq (1/2)^m |S(v,t)|$. Thus the node $v_t$ links to $v_i$ with probability at least 
$$
p(2^{-m}\rho (v)+(1-2^{-m})\rho_{\min})|S(v_i,t)|.
$$ 
Using this, we can use the exact same approach to bound the expectation of  $Y(v_i,t)$ from below. This gives the lower bound.
\end{proof}

\subsubsection*{Proof of Theorem~\ref{degconc}}

Here we show that, once a vertex has reached an in-degree of $\omega \log n$ and its area of influence is well contained within the region, its degree can be closely predicted with high probability. We will be using the following version of the Chernoff bound, as seen in e.g.~\cite[p.\ 27, Corollary~2.3]{JLR}. 

\begin{lemma}\label{lem:Chernoff} 
Let $X$ be a random variable that can be expressed as a sum of independent random indicator variables, $X=\sum_{i=1}^{n} X_{i}$, where $X_{i}\in\mathrm{Be}(p_{i})$ with (possibly) different $p_{i}=\mathbb{P} (X_{i} = 1)= \mathbb{E}X_{i}$. If $\varepsilon\le3/2$, then
\begin{align}\label{eq:Ch}
\mathbb{P} (|X - \mathbb{E }X| \ge\varepsilon\mathbb{E }X)  &  \le2 \exp\left(  - \frac{\varepsilon^{2} \mathbb{E }X}{3} \right)  .
\end{align}
\end{lemma}

\bigskip

Let us start with the following key lemma. 

\begin{lemma}\label{thm:conc}
Let $\omega=\omega(n)$ be any function tending to infinity together with $n$, and let $\eps >0$.  For a given node $v$, suppose that $\deg^-(v,T)=d \ge \omega \log n$ and that 
$$
\delta(v)\geq (1+\eps)\left(\frac{A_1d+A_2}{c_mT}\right)^{1/m}=(1+\eps)r(v,T).
$$ 
Then, with probability $1-O(n^{-3})$, for every value of $t$, $T \le t  \le 2T$, 
$$
\left| \deg^-(v,t) - d \cdot \left( \frac tT \right)^{p \rho(v) A_1} \right| \le \frac {3}{p \rho(v) A_1} \cdot \frac {t}{T} \sqrt{d \log n}.
$$
\end{lemma}

\begin{proof}
Let $\rho=\rho(v)$. Our goal is to estimate $\deg^-(v,t) - d \cdot \left( t/T \right)^{p \rho A_1}$. We will show that the upper bound holds; the lower bound can be obtained by using an analogous, symmetric, argument. Note that the assumption on $\delta(v)$ implies that $S(v,T) \subseteq \reg(v)$.

We use the following stopping time 
$$
T_0=\min \left\{ t \ge T : \deg^-(v,t) >  d \cdot \left( \frac tT \right)^{p \rho A_1} + \frac {3}{p \rho A_1} \cdot \frac {t}{T} \sqrt{d \log n} ~~\vee~~ t=2T+1 \right\}.
$$ 
Note that if $T_0=2T+1$, then the in-degree of $v$ remained bounded as required during the entire time interval $T \le t \le 2T$. Hence, in order to prove the bound, we need to show that with probability $1-O(n^{-3})$ we have $T_0=2T+1$.

Suppose that $T_0 \le 2T$. Note that for $t \ge T$ up to and including time-step $T_0-1$, the random variable $\deg^-(v,t)$ is (deterministically) bounded from above. Moreover, it is straightforward to see that this upper bound, together with the assumption on $\delta(v)$ (note the additional multiplicative $(1+\eps)$ term), implies that $S(v,t) \subseteq \reg(v)$ for all $T \le t \le T_0-1$. Hence, the number of new neighbours accumulated during this phase of the process, $\deg^-(v,T_0)-d$, can be (stochastically) bounded from above by the sum $X = \sum_{t=T}^{T_0-1} X_t$ of independent indicator random variables $X_t$, where 
$$
\Prob(X_t = 1) = p \rho \ \frac {A_1 \left( d \left( \frac tT \right)^{p \rho A_1} + \frac {3}{p \rho A_1} \cdot \frac {t}{T} \sqrt{d \log n} \right) + A_2}{t}.
$$
Clearly, since  $p \rho A_1 < 1$,
\begin{eqnarray*}
\E X &=& \sum_{t=T}^{T_0-1} \E X_t \\
&=& p \rho A_1 d T^{-p \rho A_1} \left( \sum_{t=T}^{T_0-1} t^{p \rho A_1-1} \right) + \frac {T_0-T}{T}  3 \sqrt{d \log n} + O(1) \\
&=& d \left( \frac {T_0}{T} \right)^{p \rho A_1} - d\left( \frac {T}{T} \right)^{p \rho A_1} + \frac {T_0-T}{T}  3 \sqrt{d \log n} + O(1)\\
&=& d \left( \frac {T_0}{T} \right)^{p \rho A_1} - d + \frac {T_0-T}{T}  3 \sqrt{d \log n} + O(1).
\end{eqnarray*}
Since $T_0 \leq 2T$, the in-degree of $v$ at time $T_0$ failed the desired condition, which implies that
\begin{eqnarray*}
X &\ge& \deg^-(v,T_0) - d \\
&\ge& \left( d \cdot \left( \frac{T_0}{T} \right)^{p \rho A_1} + \frac {3}{p \rho A_1} \cdot \frac {T_0}{T} \sqrt{d \log n} \right) - d \\
&=& \E X + \frac {3}{p \rho A_1} \cdot \frac {T_0}{T} \sqrt{d \log n} - \frac {T_0-T}{T}  3 \sqrt{d \log n} + O(1) \\
&\ge& \E X + 3 \sqrt{d \log n}, 
\end{eqnarray*}
using again that it is assumed that $p \rho A_1 < 1$. It follows from the Chernoff bound~(\ref{eq:Ch}) that
$$
\Prob (|X - \E X| \ge 3 \sqrt{d \log n}) \le 2 \exp \left(  - \ \eps \sqrt{d \log n} \right),
$$
where $\eps = 3 \sqrt{d \log n} / \E X$. The maximum value of $\E X$ corresponds to $T_0=2T$ and so 
\begin{eqnarray*} 
\E X &\leq& d \left( \frac {2T}{T} \right)^{p \rho A_1} - d + \frac {2T-T}{T}  2 \sqrt{d \log n} + O(1) \\
&=& d (2^{p \rho A_1} - 1)(1+o(1)) \\
&\le& d. 
\end{eqnarray*}
So $\eps \ge 3 \sqrt{d^{-1} \log n}$. Therefore, the probability that $T_0 \le 2T$ is at most $2 \exp (- 3 \log n)$ and the proof is finished.
\end{proof}

\bigskip

Now, with Lemma~\ref{thm:conc} in hand we can get Theorem~\ref{degconc}. 

\begin{proof}[Proof of Theorem~\ref{degconc}]
Let $\omega = \omega(n)$ be a function going to infinity with $n$, and let $\eps >0$. Let $v$ be a vertex with final degree $k\geq \omega\log n$, let $\rho =\rho (v)$, and assume that $\delta = \delta (v) \ge (1+\eps)r(v,n)$. Let $\hat{T}_v$ be the first time  that the in-degree of $v$ exceeds $(\omega/2) \log n$, and $\hat{t}_v$ be the first time $t$ that the radius of influence $r(v,t)\leq \delta (1+\eps/2)^{-(1-p\rho A_1)/m}$. Moreover, let $T = \max\{\hat{T_v}, \hat{t}_v\}$ be the first time that the two events hold. Finally, let $d = \deg^-(v,T)$. We obtain from Lemma~\ref{thm:conc} that, with probability $1-O(n^{-3})$, 
$$
d \left( \frac tT \right)^{p\rho A_1} \left(1 - \frac {3}{p \rho A_1} \sqrt{d^{-1} \log n} \right) \le \deg^-(v,t) \le d \left( \frac tT \right)^{p\rho A_1} \left(1 + \frac {3}{p\rho A_1} \sqrt{d^{-1} \log n} \right)
$$
for $T \le t \le 2T$. It follows that the degree tends to grow but the sphere of influence tends to shrink between $T$ and $2T$, and thus that the conditions of Lemma \ref{thm:conc} again hold at time $2T$. We can now keep applying the same lemma for times $2T$, $4T$, $8T$, $16T, \dots$, using the final value as the initial one for the next period, to get the statement for all values of $t$ from $T$ up to and including time $n$. Precisely, for $1\leq i<i_{\max}=\lfloor \log n\rfloor+1$, let $d_i=\deg^-(v,2_i T)$. Then by Lemma~\ref{thm:conc}, we have for $i>1$ that $d_i\leq d_{i-1}2^{p\rho A_1} (1+\eps_i)$, where $\eps_i=\frac {3}{p\rho A_1} \sqrt{d_{i-1}^{-1} \log n} $. Since we apply the lemma $O(\log n)$ times (for a given vertex $v$), the following statement holds with probability $1-o(n^{-2})$ from time $T$ on: for any $2^{i-1}T\leq t<2^{i} T$, we have that 
\[
\deg^-(v,t)\leq d \left( \frac tT \right)^{p\rho A_1}\prod_{j=0}^i (1+\eps_i).
\]
It remains to make sure that the accumulated multiplicative error term is still only $(1+o(1))$.
For that, let us note that 
\begin{eqnarray*}
\prod_{j=0}^i (1+\eps_i) &=& \prod_{j=1}^{i} \left( 1 + \frac{3}{p\rho A_1} \sqrt{d^{-1} 2^{-p\rho A_1j} \log n} \right) \\
&=& (1+o(1)) \exp \left( \frac{3}{p\rho A_1} \sqrt{d^{-1} \log n} \sum_{j=1}^{i} 2^{-p\rho A_1j/2} \right) \\
&=& (1+o(1)) \exp \left( O(\sqrt{d^{-1} \log n}) \right) \\
&=& 1+o(1),
\end{eqnarray*}
since $d$ grows faster than $\log n$. A symmetric argument can be used to show a lower bound for the error term and so the result holds. 

It follows that we have the desired behaviour from time $T$. Precisely, for times $T\leq t\leq n$, we have that 
\begin{equation*}
\deg^-(v,t) = d \left( \frac tT \right)^{p\rho A_1}(1+o(1)),
\end{equation*}
where $d=\deg^-(v,T) \ge \deg^-(v,\hat{T}_v) \ge (\omega/2) \log n$.  As $T = \max\{\hat{T}_v, \hat{t}_v\}$,  we need to consider two cases. Suppose first that $T=\hat{T}_v$. Setting $t=n$ and $\deg^-(v,n)=k$, we obtain that 
\[
T=(1+o(1))\left(\frac{d}{k}\right)^{1/p\rho A_1}n=(1+o(1))\left(\frac{\omega\log n}{2k}\right)^{1/p\rho A_1}n=(1+o(1)) \left(\frac12\right)^{1/p\rho A_1} T_v. 
\]
Therefore, for large enough $n$, we have that $T<T_v\leq \max\{ t_v, T_v\}$.  Suppose then that  $T = \hat{t}_v$. By definition,
\[
r(v,T) = (1+o(1)) \delta (1+\eps/2)^{-(1-p\rho A_1)/m}
\]
and, since $d>\omega \log n$,
\begin{eqnarray*}
r(v,T) &=& (1+o(1)) \left( \frac {A_1 d}{c_m T} \right)^{1/m} \\
&=& (1+o(1)) \left( \frac {A_1 k}{c_m T^{1-p\rho A_1} n^{p\rho A_1}} \right)^{1/m}, 
\end{eqnarray*}
and so
$$
T = (1+o(1)) (1+\eps/2) \left( \frac {A_1 k}{\delta^m c_m n^{p\rho A_1}} \right)^{1/(1-p\rho A_1)} = (1+o(1)) \ \frac {1+\eps/2}{1+\eps} \ t_v.
$$
Again, for large enough $n$, we have that $T<t_v\leq \max\{ t_v, T_v\}$. In either case, $T < \max\{ t_v, T_v\}$. As a result, we obtain that, for $\max\{ t_v, T_v\} \leq t\leq n$,
\begin{equation*}
\deg^-(v,t) = k \left( \frac{t}{n} \right)^{p\rho A_1}(1+o(1)).
\end{equation*}
Finally, since the statement holds for any vertex $v$ with probability $1-o(n^{-2})$, with probability $1-o(n^{-1})$  the statement holds for all vertices. The proof of the theorem is finished.
\end{proof}

Let us note that Theorem~\ref{degconc} immediately implies the following two corollaries. 
\begin{cor}\label{cor:tT}
Let $\omega =\omega (n)$ be any function tending to infinity together with $n$, and let $\eps >0$. The following holds with probability $1-o(n^{-1})$. For every node $v$, and for every time $T$ so that $\deg^-(v,T)\geq \omega\log n$ and $(1+\eps)r(v,T)\leq \delta(v)$, for all times $t$, $T\leq t\leq n$,
\[
\deg^-(v,t) = \deg^-(v,T) \left( \frac tT \right)^{p\rho (v) A_1}(1+o(1)).
\]
\end{cor}

\begin{cor}\label{cor:upper}
Let $\omega =\omega (n)$ be any function tending to infinity together with $n$. The following holds with probability $1-o(n^{-1})$. For any node $v_i$ born at time $i \ge 1$, and $i \le t \le n$ we have that
\begin{equation}\label{upperbounddegree}
\deg(v_i,t)\leq \omega \log n \left(\frac{t}{i}\right)^{p \rho_{\max}A_1}.
\end{equation}
\end{cor}
\begin{proof}
The statement is trivially true if $\deg(v_i,t) \le \omega \log n$ so we may assume that $t$ is such that $\deg(v_i, t) > \omega \log n$. Let $T$ be the first time that the in-degree of $v$ exceeds $(\omega/2) \log n$; clearly, $\deg^-(v,T) = (1/2+o(1)) \omega \log n$. First assume that $(1+\eps)r(v,T)< \delta(v)$ for some $\eps > 0$. By Corollary~\ref{cor:tT}, 
\[
\deg^-(v,t) = \deg^-(v,T) \left( \frac tT \right)^{p\rho (v) A_1}(1+o(1))\leq \omega \log n \left(\frac{t}{i}\right)^{p \rho_{\max}A_1},
\]
where the last step follows since we may assume that $n$ large enough so that the $(1+o(1))$ term is less than $2$, and the fact that $T>i$ and $\rho\leq \rho_{\max}$.

However, even if $v$ is close to the border of the region, that is, $(1+o(1))r(v,T) \geq \delta(v)$, this argument applies. Namely, in this case the degree of $v$ is stochastically bounded above by the degree of a node with the same birth time, but born in a region with density $\rho_{\max}$, and position far from the border. Therefore, the argument above still applies. 
\end{proof}

\subsection*{Proof of Theorem~\ref{thm:noEdges}}
 
Let $Z_j$ be the indicator variable of the event $\{ v_j\in \reg_\ell\}$. By definition of the process, $Z_j$ is a Bernouilli variable with expectation $q_{\ell}$, and the variables $Z_j$ are independent for different values of $j$. Thus, $|V(G_n)\cap \reg_\ell |=\sum_{j=1}^n Z_j$ is the sum of $n$ independent binomial random variables. Thus $\E(|V(G_n)\cap \reg_\ell |)=q_{\ell}n$, and it follows from Lemma~\ref{lem:Chernoff} that for every $\ell$, \aas\ $\left| |V(G_n)\cap \reg_\ell |-q_{\ell}n\right| = O(\omega \sqrt{n}) = o(n)$, where $\omega = \omega(n)$ is any function tending to infinity together with $n$. Since the number of regions is assumed to be a constant independent of $n$, the desired property holds \aas\ for all regions.

For the second statement, we examine the number of edges whose endpoints are in the same region, that is, the number of edges that do not cross a boundary. We set $M^{\ell}_t$ to be $\bigl|\{(u,v) \in E(G_t) ~|~ u,v \in \reg_{\ell}\}\bigr|$. We create the indicator variable $X_{i,j}$, such that 
$$
X_{i,j} =
\begin{cases}
      1, & \text{if }v_i, v_j \in \reg_{\ell} \text{ and }(v_i,v_j) \in E(G_n)\\
      0, & \text{otherwise}.
   \end{cases}
$$
Thus,
\begin{align}
\E (M^{\ell}_{t+1} | G_t) &= 
M^{\ell}_t + \sum_{v_j \in \reg_{\ell}, j \leq t} \E (X_{t+1,j}| G_t)  \nonumber\\ 
&=M^{\ell}_t + \sum_{v_j \in \reg_{\ell}, j \leq t} p\rho_\ell |S(v_j,t)\cap \reg_\ell |\nonumber \\
&=M^{\ell}_t + \sum_{v_j \in \reg_{\ell}, j \leq t} p\rho_\ell |S(v_j,t)|-p\rho_\ell Y^{\ell}_t,
\label{basicEdges}
\end{align}
where 
\[
Y^{\ell}_t=\sum_{v_j \in \reg_{\ell}, j \leq t}  |S(v_j,t)\setminus \reg_\ell |,
\]
Using the definition of $S(v_j,t)$, we obtain 
\begin{align}
\E (M^{\ell}_{t+1} | G_t) &=
 M^{\ell}_t + \sum_{v_j \in \reg_{\ell}, j \leq t} p\rho_{\ell} \frac{A_1\deg^-(v_j,t) + A_2}{t}-p\rho_{\ell}Y^{\ell}_t\nonumber \\
&=  M^{\ell}_t + \frac{p\rho_{\ell}A_1 (M^{\ell}_t +Z^{\ell}_t)}{t} + \frac{p\rho_{\ell}A_2}{t}|V_t\cap\reg_\ell|-p\rho_{\ell}Y^{\ell}_t,
\label{basicEdges2}
\end{align}
where 
\[
Z^{\ell}_t= | \{ (u,v)\in E(G_t) \,|\, v\in \reg_\ell,\, u\not\in\reg_\ell\}|.
\]

Let $\alpha=p\rho_{\max} A_1$. By Corollary \ref{cor:upper}, for any function $\omega = \omega(n)$ tending to infinity together with $n$, with probability $1-o(n^{-1})$, for all vertices $v_i$ and for all times $i \le t \le n$,
\[
\deg^-(v_i,t)\leq \omega\log n\left(\frac{t}{i}\right)^{\alpha}. 
\]
Let $\mathcal G$ be the event that we have these upper bounds for all vertices $v_i$ and for all times $i \le t \le n$. 

\medskip

Assume first that $\mathcal G$ holds. It follows (deterministically) that for all vertices $v_i$ and for all times $i \le t \le n$,
\[
|S(v_i,t)|=\frac{A_1\deg^-(v_i,t)+A_2}{t} \leq (\omega^2 \log n) t^{\alpha -1} i^{-\alpha},
\]
and thus  
\[
r(v_i,t)\leq (c_m^{-1/m} \omega^{2/m} \log^{1/m} n)t^{(\alpha -1)/m} i^{-\alpha/m} \leq (\omega^3 \log^{1/m} n ) t^{(\alpha -1)/m} i^{-\alpha/m}, 
\]
as usual, assuming that $n$ is large enough.
Let $X_i$ be the indicator variable of the event $\{ S(v_i,t)\not\subseteq \reg_\ell\}$. 
Using the bound on $|S(v_i,t)|$, we obtain
\[
Y^{\ell}_t\leq \sum_{ i \leq t}  X_i |S(v_i,t) |\leq \sum_{i\leq t}   X_i(\omega^2 \log n) t^{\alpha -1} i^{-\alpha}.
\]
If $X_i=1$ then $v_i$ has distance at most  $(\omega^3 \log^{1/m} n ) t^{(\alpha -1)/m} i^{-\alpha/m}$ from the boundary of $\reg_\ell$. Since the length of the boundary of $\reg_\ell$ is at most 4, the boundary strip  in which $v_j$ must be located has area at most $(4 \omega^3 \log^{1/m} n ) t^{(\alpha -1)/m} i^{-\alpha/m}$. Thus
\[
\E (X_i|{\mathcal G})\leq (4 \omega^3 \log^{1/m} n ) t^{(\alpha -1)/m} i^{-\alpha/m}.
\]
Combining this with the bound on $Y^{\ell}_t$, we obtain
\begin{eqnarray*}
\E (Y^\ell_t |{\mathcal G}) &\leq &  \sum_{i\leq t}   (4 \omega^3 \log^{1/m} n ) t^{(\alpha -1)/m} i^{-\alpha/m} (\omega^2 \log n)t^{\alpha -1} i^{-\alpha}\\
&=&(4 \omega^5 \log^{1+1/m} n)t^{-(1-\alpha)(1+\frac{1}{m})} \sum_{i\leq t} i^{-\alpha(1+\frac{1}{m})}\\
&\leq & c (\omega^5 \log^{2+1/m} n) t^{-\beta},  
\end{eqnarray*}
where  $\beta =\min\{\frac{1}{m},(1-\alpha)(1+\frac{1}{m}) \}$ and $c$ is an appropriate constant independent of $t$ and $n$. Note that $\beta >0$, and thus, for $t\ge \log^{4/\beta}n$, the expectation of $Y^\ell_t$ is $o(1)$ and so negligible (as $\omega$ can be tending to infinity arbitrarily slowly). 

Next, consider the number of cross-border edges, $Z^\ell_t$. At time $t$, an edge from $v_t$ to $v_i$, $i<t$,  which contributes to $Z^\ell_t$ can only be created if $v_t\in  S(v_i,t)\setminus \reg_\ell$. The probability that such an edge is created is at most $p\rho_{\max} |S(v_i,t)\setminus \reg_\ell| $. Thus we have that 
\[
\E (Z^{\ell}_{t}|Z_{t-1}^{\ell})=Z_{t-1}^{\ell}+\sum_{v_i \in \reg_\ell} p\rho_{\max}|S(v_i,t)\setminus \reg_\ell| =Z_{t-1}^{\ell}+p\rho_{\max} Y^{\ell}_t.
\] 
Therefore, 
\[
\E (Z^{\ell}_{t}|{\mathcal G})=\sum_{\tau =1}^t p\rho_{\max} \E (Y^{\ell}_{\tau}|{\mathcal G})\leq c' (\log^4 n) t^{1-\beta},
\]
for some constant $c'$  independent of $t$ and $n$.

By taking the expectation of (\ref{basicEdges2}) and setting $m^{\ell}_t=\E(M^{\ell}_t|{\mathcal G})$, we obtain the following recurrence for the (conditional) expected value $m^{\ell}_t$ and $t \geq \log^{1+4/\beta} n$,
\begin{equation}\label{eq:mt}
m^{\ell}_{t+1} = m^{\ell}_t\left(1 + \frac{p\rho_{\ell}A_1}{t}\right
) + p\rho_{\ell}q_\ell A_2 + o(1) .
\end{equation}
To solve this recurrence, we use the following lemma on real sequences, which is Lemma~3.1 from~\cite{fanchung}.
\begin{lemma}\label{lem}
If $(\alpha_{t}),$ $(\beta_{t})$ and $(\gamma_{t})$ are real sequences satisfying the relation
\[
\alpha_{t+1}=\left( 1-\frac{\beta_{t}}{t}\right) \alpha_{t}+\gamma_{t},
\]
and $\lim_{t\rightarrow \infty }\beta_{t}=\beta>0$ and $\lim_{t\rightarrow \infty }\gamma_{t}=\gamma,$ then $\lim_{t\rightarrow \infty}\frac{\alpha_{t}}{t}$ exists and equals $\frac{\gamma}{1+\beta}$.
\end{lemma}
Using this lemma, we see that $\lim_{t\rightarrow \infty} \frac{m_t^{\ell}}{t}=\frac{p\rho_{\ell}A_2}{1-p\rho_{\ell}A_1}q_\ell$, and thus the (first-order) solution of the recurrence (\ref{eq:mt}) is:
\begin{align}
m^{\ell}_n = (1+o(1))\frac{p\rho_{\ell}A_2}{1-p\rho_{\ell}A_1}q_\ell n .\label{upBoundEdge}
\end{align}
Here we use that $p\rho_{\ell}A_1 < 1$ as given. Note that the $o(1)$ term in the recurrence and the lower bound on $t$ only affect the $(1+o(1))$ term of the solution. 

\medskip

Finally, since $\Prob (\bar{\mathcal G}) =o(n^{-1})$ and (deterministically) $M^{\ell}_n \le {n \choose 2}$, we get  
\[
\E (M^{\ell}_n)=\Prob ({\mathcal G})m_n^{\ell}+\Prob (\bar{\mathcal G})\E (M^{\ell}_n|\bar{\mathcal G})= (1+o(1)) m_n^{\ell}+o(n)=(1+o(1))\frac{p\rho_{\ell}A_2}{1-p\rho_{\ell}A_1}q_\ell n.
\]
This completes the proof.

\subsection*{Proof of Theorem~\ref{thm:cn2}}

Fix $\omega$ and $\eps$ as in the statement of the theorem. 
Let $u$ and $v$ be nodes of final degrees $\deg^{-}(u,n)=k$ and $\deg^{-}(v,n)=j$ such that $k \geq j \ge \omega^2 \log n$,  where both $u$ and $v$ are located in a region $\reg$ with density $\rho$. Let $t_v$ and $T_v$ be defined as in~(\ref{eqn:tv}). 
Assume that the conditions of the theorem hold, that is,
\begin{equation}\label{eqn:condition}
\delta(v)^m \geq cj\mbox{ and }\delta (u)^m\geq ck,\mbox{ where }
c=(1+\eps) \left(\frac{A_1}{c_m n^{p\rho A_1}T_{v}^{1-p\rho A_1}}\right).
\end{equation}

Since $T_v = n\left(\omega \log n/j\right)^{1/p\rho A_1}\leq n\omega^{-1/p\rho A_1}=o(n)$, it follows that 
$$
\delta(v) \ge (1+\eps)\left(\frac{A_1 j + A_2}{c_m n}\right)^{1/m} \ \ \mbox{ and } \ \ \delta(u) \ge (1+\eps)\left(\frac{A_1 k + A_2}{c_m n}\right)^{1/m}.
$$
Therefore, the conditions of Theorem~\ref{degconc} are satisfied. Thus, from time $\max\{ T_{v},t_v\}$ until the end of the process we have concentration of the degree of node $v$, as given by~(\ref{eqn:deg}). Recall that $t_{v} = (1+\eps)\left(\frac{A_1 j}{c_m n^{p\rho A_1}\delta(v)^m}\right)^{1/(1-p\rho A_1)}$ as in~(\ref{eqn:tv}). Rewriting condition (\ref{eqn:condition}), we obtain that 
\[
T_{v} \geq (1+\eps)^{1/(1-p\rho A_1)}\left(\frac{A_1 j}{c_m n^{p\rho A_1}\delta(v)^m}\right)^{1/(1-p\rho A_1)} >t_v.
\]
Thus $\max\{ T_{v},t_v\} = T_v$ and so, in fact, the degree of $v$ is concentrated from time $T_v$ until time $n$.

Similarly, let $t_u$ and $T_u$ be defined as in~(\ref{eqn:tv}). The argument above, with $j$ replaced by $k$, shows that $\max\{ t_u,T_u\}=T_u$. By definition, as $k \ge j$, it follows that $T_u\leq T_v$. Hence, by Theorem~\ref{degconc}, we have concentration of the degrees of both $u$ and $v$ from time $T_v$ on.
Precisely, we have that, with probability $(1-o(n^{-1}))$, for all pairs $u$ and $v$ satisfying the conditions as stated above, and for all times $t$, $T_v\leq t\leq n$,
\begin{equation}\label{eqn:conc}
\deg^{-}(u,t) = (1+o(1))k\left(\frac{t}{n}\right)^{p\rho A_1} \ \ \ \mbox{ and }\ \ \ \deg^{-}(v,t) = (1+o(1))j\left(\frac{t}{n}\right)^{p\rho A_1}.
\end{equation}

The rest of the proof proceeds as the proof of Theorem 3.1 in~\cite{geo1}.
\paragraph{Case 1.}
By (\ref{eqn:conc}) above and the definition of $T_v$ and $T_u$, we have that $\deg^{-} (v,T_v)=(1+o(1))\omega \log n$, and  $\deg^{-}(u,T_{v}) = (1+o(1)) (\omega \log n)(k/j)$. 
This implies that the radius of influence of $u$ satisfies:
$$
r(u,T_{v}) = \Theta \left( \left(\frac{(\omega \log n)(k/j) }{T_{v}}\right)^{1/m} \right). 
$$
The condition assumed in this case is that
$$
d(u,v) \ge \eps \left(\frac{(\omega \log n) (k/j)}{T_{v}}\right)^{1/m}.
$$ 
Thus, at time $T_v$, both radii $r(u,T_{v})$ and $r(v,T_{v})$ are $O(d(u,v))$. Moreover, both radii decrease (in order) from time $T_{v}$ onwards, whereas $d(u,v)$ is independent of time. Thus there must exist a constant $c$ (depending on $\eps$ but not on $n$) such that, for all times $t$, $cT_{v}\leq t\leq n$, the areas of influence of $u$ and $v$ are disjoint.  Then the total number of common neighbours could only be, at most, $\min \{ \deg^-(u,cT_{v}), \deg^-(v,cT_{v}) \} = O(\omega \log n)$. 

\paragraph{Case 2.} 
Suppose $d$ satisfies
$$
d(u,v) \le \left( \frac {A_1 k+A_2}{c_m n} \right)^{1/m} - \left( \frac {A_1 j + A_2}{c_m n} \right)^{1/m} =r(u,n)-r(v,n).
$$
Note that this condition implies (deterministically) that at time $n$ the sphere of influence of $v$ is contained in the sphere of influence of $v$. We now see that this situation occurs in approximate form throughout the process. From~(\ref{eqn:conc}), we see that, for $T_v\leq t\leq n$, $\deg(v,t)=(1+o(1))\left(\frac{j}{k}\right)\deg(u,t)$, and thus $r(v,t)=(1+o(1))\left(\frac{j}{k}\right)^{1/m}r(u,t)$. 

If $j\leq \alpha k$ for some constant $\alpha \in (0,1)$, then we have that the difference between $r(v,t)$ and $r(u,t)$ behaves as $(1+o(1)) c r(u,t)$, where $c=\left(\frac{j}{k}\right)^{1/m}\leq \alpha^{1/m}<1$. In particular, this implies that this difference tends to shrink over time. Thus we have that, for $T_v\leq t\leq n$, 
$$
d(u,v)\leq r(u,n)-r(v,n) \leq (1+o(1))(r(u,t)-r(v,t)).
$$
Therefore, all but a negligible fraction of the sphere of influence of $v$ lies inside the sphere of influence of $u$ during the whole process. 

If $k=(1+o(1))j$, so $(\frac{j}{k})^{1/m}=1+o(1)$, the difference between $r(v,t)$ and $r(u,t)$ is $o(r(u,t))$ for $T_v\leq t\leq n$, and specifically at time $t=n$. From the condition of this case, we know that at time $n$, $S(v,n)$ is completely contained inside $S(u,n)$. This, combined with the fact that $r(v,t)-r(u,t)=o(r(u,t))$ implies that the spheres of influence of $u$ and $v$ overlap in all but a negligible part during the entire process. 

Any vertex $w$ that links to $v$ must lie inside the sphere of influence of $v$. Since most of the sphere of influence of $v$ is contained in that of $u$ in both cases mentioned above, this means that it is likely that $w$ lies inside the sphere of influence of $u$ as well, and thus has a probability $p$ of also linking to $v$.  Accounting for the small variation in the size of the spheres of influence, we have that the probability that a neighbour of $u$, added between times $T_v$ and $n$,  is also a neighbour of $v$ is $(1+o(1))p$. The number of common neighbours accumulated until time $T_v$ is at most $\deg(u,T_v)=O(\omega \log n)$, which is of smaller order than $j$, the final degree of $u$. 

Therefore, $\E cn(u,v,n) = (1+o(1)) p j$. Finally, note that the number of common neighbours is a sum of independent random indicator variables with Bernouilli distribution; each variable corresponds to a situation when a neighbour of $v$ falls into a sphere of influence of $u$. It follows that $cn(u,v,n) \in \textrm{Bin}((1+o(1))j,p)$. The concentration follows from the bound (\ref{lem:Chernoff}).

\paragraph{Case 3.}  

Suppose that $d=d(u,v)$ satisfies
$$
r(u,n)-r(v,n) < d(u,v) <\eps  \left(\frac{(\omega \log n) (k/j)}{T_{v}}\right)^{1/m}= (1+o(1)) \eps A_1^{-1/m} r(u,T_v).
$$
The analysis for this case is based on the assumption that, at time $T_v$, the sphere of influence of $v$ is contained in that of $u$, and at time $n$ the spheres are disjoint. Thus, in some narrow time interval around a time $t_0$ between times $T_v $ and $n$, the spheres become separated, and after this no more common neighbours can be formed. However, the conditions of this case do not guarantee that we have this situation: it may be that the conditions hold, but the sphere of influence of $v$ is not completely contained in that of $u$ at time $T_v$, or that the spheres are not completely disjoint at time $n$. However, the conditions guarantee that the asymptotic behaviour still applies in this case.

Let $t^-$ be the first time instance after $T_v$ when $S(u,t)$ is not completely contained in $S(v,t)$. Let $t^+$ be the last time when the spheres overlap, or $t^+=n$ if the spheres overlap at time $n$. (So  $T_v \le t^- \le t^+\leq n$). From time $T_v$ until time $t^-$, each neighbour of $u$ will be a common neighbour of $v$ and $u$ with probability $p$. From time $t^+$ to $n$, no common neighbours can be created. From time $t^-$ until time $t^+$, the probability that a neighbour of $u$ becomes a neighbour of $v$ is \emph{at most} $p$.  Thus, $p \deg^-(u,t^-)$ and $p \deg^-(u,t^+)$ form a lower and an upper bound, respectively, on the expected number of common neighbours of $u$ and $v$.

Since the centres of $S(u,t)$ and $S(v,t)$ are at distance $d$ from each other, the definition of $t^-$ and $t^+$ translate into the following conditions on the radii of the spheres of influence:
\begin{eqnarray*}
r(v, t^-) - r(u,t^-) &=& (1+o(1)) d,\\
r(v, t^+) + r(u,t^+) &=& (1+o(1)) d.
\end{eqnarray*}
(The factor $(1+o(1))$ is caused by the fact that the spheres of influence increase or decrease in discrete amounts.)

Using Equation (\ref{eqn:conc}) for the degree of $u$ and $v$ from time $T_v$, and translating this into conditions on the radius of the sphere of influence, we obtain 
\begin{eqnarray*}
&r(v,t^-)-r(u,t^-)&\\
 &= (1+o(1))
 &\left(\left(\frac{A_1 k\left(\frac{t^-}{n}\right)^{p\rho A_1} +A_2}{c_m t^-}\right)^{1/m}-\left(\frac{A_1 j\left(\frac{t^-}{n}\right)^{p\rho A_1} +A_2}{c_m t^-}\right)^{1/m}\right)
\\
&=  (1+o(1)) &\left(
\frac{A_1}{c_m} kn^{-p\rho A_1}( t^-)^{p\rho A_1-1} \right)^{1/m} \left( 1-\left(\frac{j}{k}\right)^{1/m} \right).
\end{eqnarray*}
A similar argument shows that 
\[
r(v,t^+)+r(u,t^+) = (1+o(1)) \left(
\frac{A_1}{c_m} kn^{-p\rho A_1}(t^+)^{p\rho A_1-1} \right)^{1/m} \left( 1+\left(\frac{j}{k}\right)^{1/m} \right).
\] 
Define $t_0 = (t^+ + t^-)/2$. Then, by the above, we get that $t^-=t_0 (1-O(\left(j/k\right)^{1/m}) )$, $t^+=t_0 (1+O(\left(j/k\right)^{1/m}))$, and 
\[
t_0 = (1+o(1)) \left( \frac{A_1 k}{c_m} n^{-p\rho A_1} d^{-m} \right)^{\frac{1}{1-p\rho A_1}}.
\]

It follows from the discussion from Case 2 that \aas~the number of common neighbours of $u$ and $v$ is bounded from below by $(1+o(1))p\deg^-(u,t^-)$, and from above by $(1+o(1))p\deg^-(v,t^+)$. Using our knowledge about the behaviour of the in-degree of $u$ given by~(\ref{eqn:conc}), this leads to the following expression:
\begin{eqnarray*}
cn(u,v,n)&=& pj \left(\frac{t_0}{n}\right)^{p\rho  A_1} \left( 1+o(1)+O \left( \left(\frac{j}{k}\right)^{1/m} \right) \right)\\
&=&  pj n^{-\frac{p\rho A_1}{1-p\rho A_1}} \left( \frac{A_1}{c_m} kd^{-m} \right)^{\frac{p\rho A_1}{1-p\rho A_1}} \left( 1+o(1)+O \left(\left(\frac{j}{k}\right)^{1/m} \right) \right).
\end{eqnarray*}

This concludes the proof. 

\bibliographystyle{plain}      
\bibliography{bibliographyLump}

\end{document}